\documentclass{ws-m3as}
\usepackage[superscript]{cite}
\usepackage{color}
\begin{document}

\markboth{L. Liu et al.}{Evolutionary dynamics of cooperation in a population}

%
\catchline{}{}{}{}{}
%

\title{Evolutionary dynamics of cooperation in a population with probabilistic corrupt enforcers and violators}

\author{Linjie Liu}
\address{School of Mathematical
Sciences, University of Electronic Science and Technology of China,
Chengdu 611731, China}

\author{Xiaojie Chen*}
\address{School of Mathematical
Sciences, University of Electronic Science and Technology of China,
Chengdu 611731, China\\
xiaojiechen@uestc.edu.cn\\ $*$Corresponding author}

\author{Attila Szolnoki}
\address{Institute of Technical Physics and Materials Science, Centre for Energy Research, Hungarian Academy of Sciences, P.O. Box 49, H-1525 Budapest, Hungary\\szolnoki.attila@energia.mta.hu
}

\maketitle


\begin{abstract}
Pro-social punishment is a key driver of harmonious and stable society. However, this institution is vulnerable to corruption since law-violators can avoid sanctioning by paying bribes to corrupt law-enforcers. Consequently, to understand how altruistic behavior survives in a corrupt environment is an open question. To reveal potential explanations here we introduce corrupt enforcers and violators into the public goods game with pool punishment, and assume that punishers, as corrupt enforcers, may select defectors probabilistically to take a bribe from, and meanwhile defectors, as corrupt violators, may select punishers stochastically to be corrupted. By means of mathematical analysis, we aim to study the necessary conditions for the evolution of cooperation in such corrupt environment. We find that cooperation can be maintained in the population in two distinct ways. First, cooperators, defectors, and punishers can coexist by all keeping a steady fraction of the population. Second, these three strategies can form a cyclic dominance that resembles a rock-scissors-paper cycle or a heteroclinic cycle. We theoretically identify conditions when the competing strategies coexist in a stationary way or they dominate each other in a cyclic way. These predictions are confirmed numerically.
\end{abstract}



\section{Introduction}

Understanding the emergence and persistence of altruistic behavior among selfish individuals has long been an enormous challenge for research community~\cite{fehr2003nature,Szolnoki2007Cooperation,santos2008social,santos2012role,Dolfin14,Perc2017Statistical,Allen18,chen18,szolnokipre18}.
This conundrum of collective behaviors has been often studied in social-economical systems where different mathematical tools including the kinetic theory of active particles and agent based modeling can be used~\cite{Hamilton1963The,Bellomo09,Bellomo11,Pareschi2013,Herrero15,Marsan16,Burini16,Bellomo16,Dolfin17,Bellomo17}. As an alternative approach, evolutionary game theory has also been applied to model individuals interactions at a microscopic scale for understanding the dynamics of collective actions of cooperation in large systems of interacting entities~\cite{Tanimoto07,Szolnoki2018,Perc2008Social,Perc2010Coevolutionary,fu2010invasion,Santos12,Vasconcelos13,Vasconcelos15}. Particularly, public goods game (PGG), as a standard metaphor of the mentioned social dilemma, has attracted a lot of attention from broad range of research disciplines~\cite{milinski2002reputation,hauert2002volunteering,Szolnoki2009Resolving,hilbe18,he2018amc,Wu18}. In the PGG, individuals decide whether to contribute to the common pool or not, and the accumulated and enhanced contribution are distributed to group members equally. Thus defectors who contribute nothing will obtain a higher payoff, no matter what other players do. If left controlled, each individual prefers to defect in the game leading to the collapse of cooperation~\cite{hardin1968tragedy}.

In order to solve the problem of cooperation in the game, an intensive research efforts have been carried out during the past few decades~\cite{nowak2006five,Fu2008Reputation,santos2011risk,Szolnoki2012Evolutionary,chen2012risk,Sasaki2012The}.
One prominently discussed solution is the employment of pro-social punishment, that is, punishing uncooperative individuals by lowering their income~\cite{rockenbach2006efficient,Szolnoki2011Phase,sasaki2013evolution,Szolnoki2013Effectiveness,chen2014probabilistic,Chenxj2015,Liu2017Competitions,szolnoki2017alliance,Liu2018Chaos,Yang2018}. In particular, two different ways of punishment are studied, namely, peer punishment and pool punishment. The former refers to that group individuals can impose fines on the violators directly and its additional cost~\cite{Fehr02,Andreoni03,rockenbach2006efficient}. The latter means that players can decide whether to contribute to a punishment pool before contributing to the common pool of the basic game and sanctioning defection is organized centrally~\cite{Sigmund2010nature,Szolnoki2011Phase}. This way of punishment is widespread and generally preferred in modern human societies~\cite{Sigmund2010nature}.

While punishment can solve the original dilemma of cooperation, the effectiveness of punishment in promoting cooperation has been challenged by recent theoretical research showing that the existence of corruption where defectors bribe corrupt officials to avoid punishment can destroy the positive consequence of costly punishment in cooperation~\cite{Abdallah14,Verma2015,Lee2017,Muthukrishna2017,Huang2018}.
For example, Muthukrishna et al.~\cite{Muthukrishna2017} experimentally showed that the possibility of corruption can cause a significant fall in public good provisioning and make empowering leaders decrease cooperative contributions.

It is worth mentioning that in most of the studies involving corruption, it is always assumed that defectors bribe corrupt enforcers and meanwhile enforcers who are involved with corruption take the bribes permanently. Indeed this assumption is not always justified in realistic situations where officials who are involved in corruption do not accept all the offered bribes and meanwhile civilians who violate the rules are not always going to offer a bribe to the officials. In other words, corrupt officials and defective civilians may act stochastically and offer or take bribes occasionally.

Motivated by the above mentioned observations, in this work we thus introduce corrupt enforcers and violators into the PGG with pool punishment, and we assume that defectors bribe enforcers probabilistically to avoid punishment and meanwhile pool punishers accept a bribe from defectors also stochastically. By applying the replicator equation approach we assume infinite population where the evolutionary game dynamics is used. Our theoretical analysis reveals, which is also confirmed by numerical calculations, that the cooperative behavior can be maintained in the population in two different ways. First, cooperators, defectors, and punishers can coexist where the portions of all strategies are stable in time. Alternatively, the three strategies can dominate each other in a cyclic way similarly to the well known rock-scissors-paper cycle or the heteroclinic cycle.

\section{Model and method}\label{section2}
\subsection{Public goods game}
We consider the PGG played in an infinite well-mixed population. At each time step, $N$ individuals are randomly chosen to form a group to participate in the PGG. Each individual can choose to cooperate or defect. Cooperators ($C$) contribute to the common pool at a cost $c$, while defectors ($D$) contribute nothing. The sum of contributions is multiplied by a factor $r$ ($1<r<N$) and then the enhanced amount is divided equally among all group members.

We then introduce a third strategist, pool punishers ($P$), who contribute $G$ to the punishment pool before contributing $c$ to the common pool. They are responsible for monitoring the entire population and punishing those individuals who are unwilling to contribute to the common pool. Alternatively, some defectors might be tempted to bribe the enforcers, so as to enable them to avoid punishment.
In the same way, some enforcers can choose to accept bribes and do not punish the bribers. Here we assume that punishers decide to accept the bribe $b$ from corrupt defectors with probability $p$.
With $1-p$ probability they do not accept any bribes and punish defectors in the group. On the other hand, we assume that defector will pay the cost $h$ to bribe the corrupt punishers with probability $q$, while with $1-q$ probability they are unwilling to do this. Thus, a defector who does not bribe will receive a fine $B (B\geq h)$ from every punisher no matter whether he/she is willing to accept bribes.

As a result, the payoffs of cooperators, defectors, and punishers from one PGG can be respectively
written as
\begin{eqnarray}
\pi_{C}&=&\frac{rc(N_{C}+N_{P}+1)}{N}-c,  \\
\pi_{D}&=&\frac{rc(N_{C}+N_{P})}{N}-(1-pq)BN_{P}-pqN_{P}h,  \\
\pi_{P}&=&\frac{rc(N_{C}+N_{P}+1)}{N}-c-G+pqN_{D}b,
\end{eqnarray}
where $N_{C}, N_{D}$, and $N_{P}$ respectively represent the number of cooperators, defectors, and pool punishers among the other $N-1$ group members, and $(1-pq)B$ denotes the expected fine from each punisher.

\subsection{Replicator equation}

We apply replicator equations to study the evolutionary dynamics of strategies in our model~\cite{Santos12,Nowak04S,hofbauer1998evolutionary,wang2018}. We denote by $x, y,$ and $z$ the frequencies of $C, D$, and $P$, respectively. Thus, $x, y, z\geq0$ and $x+y+z=1$. The replicator equations are written as
\begin{eqnarray}\label{system1}
\left\{
\begin{aligned}
\dot{x}&=x(P_{C}-\bar{P}),  \\
\dot{y}&=y(P_{D}-\bar{P}),  \\
\dot{z}&=z(P_{P}-\bar{P}),
\end{aligned}
\right.
\end{eqnarray}
where $P_{C}, P_{D}$, and $P_{P}$ denote the expected payoffs of $C, D,$ and $P$, respectively, and $\bar{P}=xP_{C}+yP_{D}+zP_{P}$ gives the average payoff of the entire population. Accordingly, the expected payoffs of the above three strategies can be written as
\begin{eqnarray}\label{equation2.5}
P_{i}=\sum_{N_{C}=0}^{N-1}\sum_{N_{D}=0}^{N-N_{C}-1}\binom{N-1}{N_{C}}\binom{N-N_{C}-1}{N_{D}}x^{N_{C}}y^{N_{D}}z^{N-N_{C}-N_{D}-1}\pi_{i},
\end{eqnarray}
where $i=C, D,$ or $P$.

In the next section, we examine the evolutionary dynamics for the mentioned three strategies. Particulary, we analyze the distribution and stability of equilibrium points in the following section.

\section{Theoretical analysis}\label{section3}

\subsection{Equilibrium points}
We consider the replicator dynamics for cooperators ($C$), defectors ($D$), and pool punishers ($P$), with the frequencies $x, y$, and $z$, respectively.
We can get the expected payoffs of these three strategies by simplifying the formula presented in Eq.~(\ref{equation2.5}) as follows:
\begin{eqnarray}
P_{C}&=&\frac{rc}{N}(N-1)(x+z)+\frac{rc}{N}-c,  \\
P_{D}&=&\frac{rc}{N}(N-1)(x+z)-(1-pq)B(N-1)z-pq(N-1)zh,\\
P_{P}&=&\frac{rc}{N}(N-1)(x+z)+\frac{rc}{N}-c-G+pq(N-1)yb,
\end{eqnarray}
where $(N-1)(x+z)$ denotes the expected numbers of contributors among the $N-1$ co-players, and $B(N-1)z$ gives the expected fine on a defector. The detailed analysis of the equilibrium points is shown in Theorem~\ref{theorem3.1}.

\begin{theorem}\label{theorem3.1}
The system described by~(\ref{system1}) has at most five equilibria $(x, y, z)=(1, 0, 0), (0, 1, 0), (0, 0, 1), (0, 1-\alpha, \alpha)$, and $(1-\beta-\mu, \mu, \beta)$. Here we introduce the abbreviation $\alpha=\frac{rc/N-c-G+pq(N-1)b}{(B+b-h)(N-1)pq-B(N-1)}$, $\beta=\frac{c-rc/N}{(1-pq)B(N-1)+pqh(N-1)}$, and $\mu=\frac{G}{pq(N-1)b}$ \label{T3.1}.
\end{theorem}
\begin{proof}
It is easy to obtain the three vertex fixed points of the above system, namely, $(0, 0, 1), (1, 0, 0)$, and $(0, 1, 0)$. In the following, we analyze the boundary and interior equilibria of the system in detail.

We first investigate the interior equilibrium points in the simplex $S_{3}$. Solving $P_{C}=P_{D}$ results in $z=\beta$. Similarly, by solving $P_{C}=P_{P}$, we have $y=\mu$. Thus, there is an interior equilibrium point $(1-\beta-\mu, \mu, \beta)$ when $\mu<1$, $\beta<1$, and $1-\beta-\mu>0$.

Then we investigate the fixed points on each edge of the simplex $S_{3}$. On the edge $CD$ we have $z=0$, resulting in $\dot{y}=y(1-y)(P_{D}-P_{C})=y(1-y)(c-\frac{rc}{N})>0$. As a result, the system evolves from the $C$ state towards the $D$ state.

On the $CP$ edge we have $y=0$, resulting in $\dot{x}=x(1-x)(P_{C}-P_{P})=x(1-x)G>0$. Thus the direction of the dynamics goes from $P$ toward $C$ state.

On the $DP$ edge we have $y+z=1$, and the replicator system changes to $\dot{z}=z(1-z)(P_{P}-P_{D})$. Solving $P_{P}=P_{D}$ results in $z=\alpha$. Thus there exists a boundary equilibrium point $(0, 1-\alpha,\alpha)$ when $0<\alpha<1$.
\end{proof}

\subsection{Stability analysis of equilibria}

Here we set that
\begin{eqnarray}
\left\{
\begin{aligned}\label{system2}
f(x,y)&=&x[(1-x)(P_{C}-P_{P})-y(P_{D}-P_{P})],\\
g(x,y)&=&y[(1-y)(P_{D}-P_{P})-x(P_{C}-P_{P})].
\end{aligned}
\right.
\end{eqnarray}
Then the Jacobian of the system is
\begin{equation}
J=\begin{bmatrix}
\frac{\partial{f(x,y)}}{\partial{x}} & \frac{\partial{f(x,y)}}{\partial{y}}\\
\frac{\partial{g(x,y)}}{\partial{x}} & \frac{\partial{g(x,y)}}{\partial{y}}
\end{bmatrix},
\end{equation}
where
\begin{eqnarray}
\left\{
\begin{aligned}
\frac{\partial{f(x,y)}}{\partial{x}}&=[(1-x)(P_{C}-P_{P})-y(P_{D}-P_{P})]+x[-(P_{C}-P_{P})\\
&+(1-x)\frac{\partial}{\partial{x}}(P_{C}-P_{P})-y\frac{\partial}{\partial{x}}(P_{D}-P_{P})],\\
\frac{\partial{f(x,y)}}{\partial{y}}&=x[(1-x)\frac{\partial}{\partial{y}}(P_{C}-P_{P})-(P_{D}-P_{P})-y\frac{\partial}{\partial{y}}(P_{D}-P_{P})],\\
\frac{\partial{g(x,y)}}{\partial{x}}&=y[(1-y)\frac{\partial}{\partial{x}}(P_{D}-P_{P})-(P_{C}-P_{P})-x\frac{\partial}{\partial{x}}(P_{C}-P_{P})],\\
\frac{\partial{g(x,y)}}{\partial{y}}&=[(1-y)(P_{D}-P_{P})-x(P_{C}-P_{P})]+y[-(P_{D}-P_{P})\\
&+(1-y)\frac{\partial}{\partial{y}}(P_{D}-P_{P})-x\frac{\partial}{\partial{y}}(P_{C}-P_{P})].
\end{aligned}
\right.
\end{eqnarray}

In the following, we study the stabilities of equilibria based on whether the system has an interior equilibrium point.

\subsubsection{The system (\ref{system1}) has an interior equilibrium point}\label{subsection3.2.1}
When $\beta<1$, $\mu<1$, and $\beta+\mu<1$, there exists an interior fixed point, namely, $(x, y, z)=(1-\beta-\mu, \mu, \beta)$. In the following we discuss the stability of this fixed point.

When $(1-pq)B+pqh-pqb<0$, the existing interior fixed point is stable.

(1) For $0<\alpha<1$, there is a boundary equilibrium point with $z=\alpha$ on the $DP$ edge. Thus the system has five fixed points in the parameter space, namely, $(0, 0, 1), (1, 0, 0), (0, 1, 0), (0, 1-\alpha, \alpha)$, and $(1-\beta-\mu, \mu, \beta)$.

For $(x, y, z)=(0, 0, 1)$, the Jacobian is
\begin{equation}
J=\begin{bmatrix}
G &&& 0\\
0 &&& -(1-pq)B(N-1)-pq(N-1)h-\frac{rc}{N}+c+G
\end{bmatrix},
\end{equation}
thus the fixed equilibrium is unstable since $G>0$.

For $(x, y, z)=(1, 0, 0)$, the Jacobian is
\begin{equation}
J=\begin{bmatrix}
-G &&& -(c+G-\frac{rc}{N})\\
0 &&& c-\frac{rc}{N}
\end{bmatrix},
\end{equation}
thus the fixed equilibrium is unstable since $r<N$.

For $(x, y, z)=(0, 1, 0)$, the Jacobian is
\begin{equation}
J=\begin{bmatrix}
\frac{rc}{N}-c &&& 0\\
-(G-pq(N-1)b) &&& \frac{rc}{N}-c-G+pq(N-1)b
\end{bmatrix},
\end{equation}
thus the fixed equilibrium is a saddle point and unstable since $\frac{rc}{N}-c-G+pq(N-1)b>0$.

For $(x, y, z)=(0, 1-\alpha, \alpha)$, the Jacobian is
\begin{equation}
J=\begin{bmatrix}
a_{11} & 0\\
a_{21} & a_{22}
\end{bmatrix},
\end{equation}
where $a_{11}=G-pq(N-1)yb, a_{21}=y(1-y)[pq(N-1)h+(1-pq)(N-1)B]-y[G-pq(N-1)yb]$, and $a_{22}=y(1-y)(N-1)[(1-pq)B+pqh-pqb]$, thus the fixed equilibrium is unstable since $G-pq(N-1)yb>0$.

For $(x,y,z)=(1-\beta-\mu, \mu, \beta)$, we define the equilibrium point as $(x^{*}, y^{*}, z^{*})$ hereafter, thus the elements in the Jacobian matrix are written as
\begin{eqnarray}
\left\{
\begin{aligned}
\frac{\partial{f}}{\partial{x}}(x^{*},y^{*})&=&x^{*}[(1-x^{*})\frac{\partial}{\partial{x}}(P_{C}-P_{P})-y^{*}\frac{\partial}{\partial{x}}(P_{D}-P_{P})],\\
\frac{\partial{f}}{\partial{y}}(x^{*},y^{*})&=&x^{*}[(1-x^{*})\frac{\partial}{\partial{y}}(P_{C}-P_{P})-y^{*}\frac{\partial}{\partial{y}}(P_{D}-P_{P})],\\
\frac{\partial{g}}{\partial{x}}(x^{*},y^{*})&=&y^{*}[(1-y^{*})\frac{\partial}{\partial{x}}(P_{D}-P_{P})-x^{*}\frac{\partial}{\partial{x}}(P_{C}-P_{P})],\\
\frac{\partial{g}}{\partial{y}}(x^{*},y^{*})&=&y^{*}[(1-y^{*})\frac{\partial}{\partial{y}}(P_{D}-P_{P})-x^{*}\frac{\partial}{\partial{y}}(P_{C}-P_{P})],
\end{aligned}
\right.
\end{eqnarray}
where
\begin{eqnarray}
\left\{
\begin{aligned}
\frac{\partial}{\partial{x}}(P_{C}-P_{P})&=0,\\
\frac{\partial}{\partial{y}}(P_{C}-P_{P})&=-pq(N-1)b,\\
\frac{\partial}{\partial{x}}(P_{D}-P_{P})&=(N-1)[pqh+(1-pq)B],\\
\frac{\partial}{\partial{y}}(P_{D}-P_{P})&=(N-1)[pqh+(1-pq)B-pqb].
\end{aligned}
\right.
\end{eqnarray}
Then we define that $p=\frac{\partial{f}}{\partial{x}}(x^{*},y^{*})\frac{\partial{g}}{\partial{y}}(x^{*},y^{*})-\frac{\partial{f}}{\partial{y}}(x^{*},y^{*})\frac{\partial{g}}{\partial{x}}(x^{*},y^{*})$ and $q=\frac{\partial{f}}{\partial{x}}(x^{*},y^{*})+\frac{\partial{g}}{\partial{y}}(x^{*},y^{*})$. Thus we have
\begin{eqnarray}
p&=&\frac{\partial{f}}{\partial{x}}(x^{*},y^{*})\frac{\partial{g}}{\partial{y}}(x^{*},y^{*})-\frac{\partial{f}}{\partial{y}}(x^{*},y^{*})\frac{\partial{g}}{\partial{x}}(x^{*},y^{*})\nonumber\\
&=&x^{*}y^{*}(1-x^{*}-y^{*})[\frac{\partial}{\partial{x}}(P_{C}-P_{P})\frac{\partial}{\partial{y}}(P_{D}-P_{P})\nonumber\\
&-&\frac{\partial}{\partial{y}}(P_{C}-P_{P})\frac{\partial}{\partial{x}}(P_{D}-P_{P})]\nonumber\\
&=&x^{*}y^{*}(1-x^{*}-y^{*})(N-1)^{2}pqb[pqh+(1-pq)B]\nonumber\\
&>&0,
\end{eqnarray}
and
\begin{eqnarray}
q&=&\frac{\partial{f}}{\partial{x}}(x^{*},y^{*})+\frac{\partial{g}}{\partial{y}}(x^{*},y^{*})\nonumber\\
&=&x^{*}(1-x^{*})\frac{\partial}{\partial{x}}(P_{C}-P_{P})+y^{*}(1-y^{*})\frac{\partial}{\partial{y}}(P_{D}-P_{P})\nonumber\\
&-&x^{*}y^{*}[\frac{\partial}{\partial{x}}(P_{D}-P_{P})+\frac{\partial}{\partial{y}}(P_{C}-P_{P})]\nonumber\\
&=&(N-1)[(1-pq)B+pqh-pqb]y^{*}(1-y^{*}-x^{*}).
\end{eqnarray}

Consequently for $(1-pq)B+pqh-pqb<0$ the interior fixed point is stable.

(2) For $\alpha\leq0$ or $\alpha\geq1$, the boundary equilibrium point of the $DP$ edge does not exist. Thus the system has four fixed points in the parameter space. The stability of $(0, 0, 1), (1, 0, 0),$ and $(1-\beta-\mu, \mu, \beta)$ will not change, compared with the case of $0<\alpha <1$. If $\alpha\geq1$, the fixed point $(0, 1, 0)$ is unstable since the largest eigenvalue of $J(0, 1, 0)$ is positive. If $\alpha<0$, the fixed point $(0, 1, 0)$ is stable. Particularly, for $\alpha=0$ we can prove that this fixed point is unstable by using the center manifold theorem~\cite{Carr81,Khalil96} (see Theorem.\ref{theorem3.3}).

\begin{theorem}\label{theorem3.3}
When $\frac{rc}{N}-c-G+pq(N-1)b=0$, the fixed point $(0, 1, 0)$ is unstable.
\end{theorem}

\begin{proof}
Because of $y = 1 - x - z$, the dynamic equations (\ref{system2}) become
\begin{eqnarray}\label{system2.2}
\left\{
\begin{aligned}
\dot{x}&=x[(1-x)(P_{C}-P_{D})-z(P_{P}-P_{D})],\\
\dot{z}&=z[(1-z)(P_{P}-P_{D})-x(P_{C}-P_{D})],
\end{aligned}
\right.\label{2.3}
\end{eqnarray}
where
\begin{eqnarray}
P_{C}-P_{D}&=&\frac{rc}{N}-c+(1-pq)B(N-1)z+(N-1)pqzh,\\
P_{P}-P_{D}&=&\frac{rc}{N}-c-G+pq(N-1)(1-x-z)b\nonumber\\
&+&(1-pq)B(N-1)z+pq(N-1)zh.
\end{eqnarray}

We know that $(x, z)=(0, 0)$ is an equilibrium point of the equation system (\ref{system2.2}). Consequently the Jacobian is
\begin{equation}
A=\begin{bmatrix}
\frac{rc}{N}-c & 0\\
0 & \frac{rc}{N}-c-G+pq(N-1)b
\end{bmatrix}.
\end{equation}
When $\frac{rc}{N}-c-G+pq(N-1)b=0$, the eigenvalues of the Jacobian for the fixed point $(x, z)=(0, 0)$ are $0$ and $\frac{rc}{N}-c$. In this condition, we study the stability of the equilibrium point by further using the center manifold
theorem~\cite{Khalil96}. To do that, we construct a matrix $M$, whose column elements are the eigenvectors of the matrix $A$, given as
\begin{equation}
M=\begin{bmatrix}
0 &&& 1\\
1 &&& 0
\end{bmatrix}.
\end{equation}
Let $T=M^{-1}$, then we have
\begin{equation}
TAT^{-1}=\begin{bmatrix}
0 &&& 0\\
0 &&& \frac{rc}{N}-c
\end{bmatrix}.
\end{equation}
Using variable substitution, we have
\begin{equation}
\begin{bmatrix}
v \\
u
\end{bmatrix}=T\begin{bmatrix}
x\\
z
\end{bmatrix}=\begin{bmatrix}
0 &&& 1\\
1 &&& 0
\end{bmatrix}\begin{bmatrix}
x\\
z
\end{bmatrix}=\begin{bmatrix}
z\\
x
\end{bmatrix}.
\end{equation}
Therefore, the system (\ref{2.3}) can be rewritten as
\begin{eqnarray}
\left\{
\begin{aligned}
\dot{v}&=v(1-v)[\frac{rc}{N}-c-G+pq(N-1)(1-u-v)b+(1-pq)B(N-1)v\nonumber\\
&+pq(N-1)vh]-vu[\frac{rc}{N}-c+(1-pq)B(N-1)v+(N-1)pqvh],\\
\dot{u}&=u(1-u)[\frac{rc}{N}-c+(1-pq)B(N-1)v+(N-1)pqvh]-vu[\frac{rc}{N}-c\nonumber\\
&-G+pq(N-1)(1-u-v)b+(1-pq)B(N-1)v+pq(N-1)vh].
\end{aligned}
\right.
\end{eqnarray}
Using the center manifold theorem~\cite{Carr81,Khalil96}, we have that $u = e(v)$ is a center manifold for the above system.
Then the dynamics on the center manifold, namely, the dynamics of
\begin{eqnarray}\label{A35}
\dot{v}&=&v(1-v)[\frac{rc}{N}-c-G+pq(N-1)(1-e(v)-v)b+(1-pq)B(N-1)v\nonumber\\
&+&pq(N-1)vh]-ve(v)[\frac{rc}{N}-c+(1-pq)B(N-1)v+(N-1)pqvh],
\end{eqnarray}
determine the dynamics near the equilibrium point. Assuming that $e(v)=O(|v|^2)$, thus the system (\ref{A35}) can be expressed as
\begin{eqnarray}
\dot{v}=(v^{2}-v^{3})[-pq(N-1)b+(1-pq)(N-1)B+pq(N-1)h]+O(|v|^{4}).
\end{eqnarray}
Since $-pq(N-1)b+(1-pq)(N-1)B+pq(N-1)h\neq0$, thus we obtain that $v=0$ is unstable. Accordingly, the fixed point $(0, 1, 0)$ is unstable in the equation system (\ref{system1}).
\end{proof}

Next we study the case when $(1-pq)B+pqh-pqb=0$. Here the eigenvalues of the Jacobian matrix for the interior equilibrium point are pure imaginary. Besides, the boundary equilibrium point of the $DP$ edge does not exist. Hence the system has four fixed points in the parameter space, namely, $(0, 0, 1), (1, 0, 0), (0, 1, 0)$, and $(1-\beta-\mu, \mu, \beta)$. The three vertexes are unstable since the largest eigenvalues of the Jacobian matrices of these three fixed points are all positive.\label{case2}

\begin{theorem} \label{theorem3.2}
When $(1-pq)B+pqh-pqb=0$, the interior fixed point is a center surrounding by periodic closed orbits.
\end{theorem}
\begin{proof} We introduce a new variable
$\varepsilon=\frac{x}{x+y}$, which represents the fraction of cooperators
among individuals who do not contribute to the punishment pool. This
yields
\begin{eqnarray}
\dot{\varepsilon}=\frac{xy}{(x+y)^{2}}(P_{C}-P_{D})=-\varepsilon(1-\varepsilon)(P_{D}-P_{C}).
\end{eqnarray}
On the other hand, $\dot{z}=z(P_{P}-\bar{P})$, where
\begin{eqnarray}
\bar{P}&=&xP_{C}+yP_{D}+zP_{P}\nonumber\\
&=&x(P_{C}-P_{D})+(1-z)(P_{D}-P_{P})+P_{P}.
\end{eqnarray}
Thus we have
\begin{eqnarray}
\dot{z}=z[x(P_{D}-P_{C})-(1-z)(P_{D}-P_{P})].
\end{eqnarray}
\setlength{\arraycolsep}{5pt}
Accordingly the equation system becomes
\begin{eqnarray}
\left\{
\begin{aligned}
\dot{\varepsilon}&=\varepsilon(1-\varepsilon)[\frac{rc}{N}-c+(1-pq)B(N-1)z+(N-1)pqzh],\\
\dot{z}&=z(1-z)\{(1-\varepsilon)[\frac{rc}{N}-c+pq(N-1)(1-z)b\nonumber\\
&+(1-pq)B(N-1)z+pq(N-1)zh]-G\}.
\end{aligned}
\right.
\end{eqnarray}
The separability of the factors allows us to write
\begin{eqnarray*}
\frac{dz}{d\varepsilon}=\frac{z(1-z)}{\frac{rc}{N}-c+(1-pq)B(N-1)z+(N-1)pqzh}\frac{(1-\varepsilon)[\frac{rc}{N}-c+pq(N-1)b]-G}{\varepsilon(1-\varepsilon)},
\end{eqnarray*}
such that
\setlength{\arraycolsep}{0.0em}
\begin{eqnarray}
\int\frac{\frac{rc}{N}-c+(1-pq)B(N-1)z+(N-1)pqzh}{z(1-z)}dz\nonumber\\
=\int\frac{(1-\varepsilon)[\frac{rc}{N}-c+pq(N-1)b]-G}{\varepsilon(1-\varepsilon)}d\varepsilon.
\end{eqnarray}
\setlength{\arraycolsep}{5pt}
The integral of the right-hand side is
\begin{eqnarray}
[\frac{rc}{N}-c+pq(N-1)b]\log(\varepsilon)-G[\log(\varepsilon)-\log(1-\varepsilon)].
\end{eqnarray}
The integral of the left-hand side is
\begin{eqnarray}
(\frac{rc}{N}-c)[\log(z)-\log(1-z)]-[(1-pq)B(N-1)\nonumber\\
+(N-1)pqh]\log(1-z).
\end{eqnarray}

In this way, we identify the constant of motion
\begin{eqnarray}
H(\varepsilon,z)&=&[\frac{rc}{N}-c+pq(N-1)b]\log(\varepsilon)-G[\log(\varepsilon)-\log(1-\varepsilon)]\nonumber\\
&+&(\frac{rc}{N}-c)[\log(z)-\log(1-z)]-[(1-pq)B(N-1)\nonumber\\
&+&(N-1)pqh]\log(1-z).
\end{eqnarray}
Therefore, we have
\begin{eqnarray}
\dot{H}=\frac{\partial{H}}{\partial{\varepsilon}}\dot{\varepsilon}+\frac{\partial{H}}{\partial{z}}\dot{z}=0.
\end{eqnarray}
Accordingly, the system is conservative and all constant level sets of $H$
correspond to closed curves. Besides, the interior fixed point is neutrally stable surrounded by those closed and periodic orbits.
\end{proof}

Last, when $(1-pq)B+pqh-pqb>0$, the existing interior fixed point is unstable.\label{case3}

(1) For $0<\alpha<1$, there is a boundary equilibrium point on the $DP$ edge. Then the system has five fixed points in the parameter space, namely, $(0, 0, 1), (1, 0, 0), (0, 1, 0), (0, 1-\alpha, \alpha)$, and $(1-\beta-\mu, \mu, \beta)$.

The fixed points $(0, 0, 1), (1, 0, 0),$ and $(0, 1-\alpha, \alpha)$ are all unstable since the largest eigenvalues of the Jacobian matrices are all positive.

For $(x,y,z)=(0,1,0)$, the Jacobian is
\begin{equation}
J=\begin{bmatrix}
\frac{rc}{N}-c &&& 0\\
-(G-pq(N-1)b) &&& \frac{rc}{N}-c-G+pq(N-1)b
\end{bmatrix},
\end{equation}
thus this is a stable equilibrium point since $\frac{rc}{N}-c-G+pq(N-1)b<0$.

(2) For $\alpha\leq0$ or $\alpha\geq1$, the boundary equilibrium point of the $DP$ edge does not exist. Thus the system has four fixed points in the mentioned parameter space. The stability of $(0, 0, 1), (1, 0, 0),$ and $(1-\beta-\mu, \mu, \beta)$ will not change, compared with the case of $0<\alpha<1$. If $\alpha\geq1$, the fixed point $(0, 1, 0)$ is stable since the largest eigenvalue of $J(0, 1, 0)$ is positive. If $\alpha\leq0$, the fixed point $(0, 1, 0)$ is unstable.

\begin{theorem}\label{theorem3.4}
When $\frac{rc}{N}-c-G>\max\{-(N-1)[(1-pq)B+pqh],-pq(N-1)b\}$ and $r<N$, there is a stable heteroclinic cycle on the boundary of the simplex $S_{3}$.
\end{theorem}

\begin{proof}
When $\frac{rc}{N}-c-G>\max\{-(N-1)[(1-pq)B+pqh],-pq(N-1)b\}$ and $r<N$,
we know that the three vertex equilibrium points ($C, D$, and $P$) are all saddle nodes, and the three edges ($CD, DP$, and $PC$) are the heteroclinic trajectories.
All of these guarantee the existence of the heteroclinic cycle on the boundary $S_{3}$. In the following, we will prove that the heteroclinic cycle is asymptotically stable.

Based on the above analysis, we can get the eigenvalues of the Jacobian matrices of the three vertex equilibrium points, namely, $\lambda_{P}^{-}=-(N-1)[(1-pq)B+pqh]-(\frac{rc}{N}-c-G), \lambda_{P}^{+}=G, \lambda_{C}^{-}=-G, \lambda_{C}^{+}=c-\frac{rc}{N}, \lambda_{D}^{-}=\frac{rc}{N}-c,$ and $\lambda_{D}^{+}=\frac{rc}{N}-c-G+pq(N-1)b$, respectively. Then we respectively define that $\lambda_{P}=-\frac{\lambda_{P}^{-}}{\lambda_{P}^{+}}, \lambda_{C}=-\frac{\lambda_{C}^{-}}{\lambda_{C}^{+}},$ and $\lambda_{D}=-\frac{\lambda_{D}^{-}}{\lambda_{D}^{+}}$,
and we have $\lambda=\lambda_{P}\lambda_{C}\lambda_{D}=\frac{\frac{rc}{N}-c-G+(N-1)[(1-pq)B+pqh]}{\frac{rc}{N}-c-G+pq(N-1)b}$. When $(1-pq)B+pq(h-b)>0$, we have $\lambda>1$.
Therefore the heteroclinic cycle is asymptotically stable~\cite{park18}.
\end{proof}

\subsubsection{There is no interior equilibrium point in the system (\ref{system1})} \label{subsection3.2.2}
When $\beta\geq1$, or $\mu\geq1$, or $\beta+\mu\geq1$, the interior fixed point does not exist.

(1) For $0<\alpha<1$, there is a boundary equilibrium point on the $DP$ edge. Then the system has four fixed points in the parameter space, namely, $(0, 0, 1), (1, 0, 0), (0, 1, 0),$ and $(0, 1-\alpha, \alpha)$.

The fixed points $(0, 0, 1)$ and $(1, 0, 0)$ are both unstable since the largest eigenvalues of Jacobian matrices are both positive.

For $(x,y,z)=(0,1,0)$, the Jacobian is
\begin{equation}
J=\begin{bmatrix}
\frac{rc}{N}-c &&& 0\\
-(G-pq(N-1)b) &&& \frac{rc}{N}-c-G+pq(N-1)b
\end{bmatrix},
\end{equation}
thus the fixed equilibrium is a saddle point and unstable when $\frac{rc}{N}-c-G+pq(N-1)b>0$; while when $\frac{rc}{N}-c-G+pq(N-1)b<0$, it is stable.

For $(x, y, z)=(0, 1-\alpha, \alpha)$, the Jacobian is
\begin{equation}
J=\begin{bmatrix}
a_{11} &&& 0\\
a_{21} &&& a_{22}
\end{bmatrix},
\end{equation}
where $a_{11}=G-pq(N-1)(1-\alpha)b, a_{21}=\alpha(1-\alpha)[pq(N-1)h+(1-pq)(N-1)B]-(1-\alpha)[G-pq(N-1)(1-\alpha) b]$, and $a_{22}=\alpha(1-\alpha)(N-1)[(1-pq)B+pqh-pqb]$, thus the fixed equilibrium is stable when $pqb>\max\{\frac{G}{(N-1)(1-\alpha)}, (1-pq)B+pqh\}$, and it is unstable when $pqb<\max\{\frac{G}{(N-1)(1-\alpha)}, (1-pq)B+pqh\}$. Particularly, when $pqb=\max\{\frac{G}{(N-1)(1-\alpha)}, (1-pq)B+pqh\}$, we find that one eigenvalue of the Jacobian matrix at this fixed point is zero and the other one is negative. Accordingly, we study its stability by further using the center manifold theorem~\cite{Khalil96} as follows.

For $0<\alpha<1$, we know $pqb\neq (1-pq)B+pqh$ and $\max\{\frac{G}{(N-1)y}, (1-pq)B+pqh\}=\frac{G}{(N-1)y}$. In order to use the center manifold theorem conveniently, we take $\xi =y-1+\alpha$, then the equation system becomes
\begin{eqnarray*}
\left\{
\begin{aligned}
\dot{x}&=x(1-x)[G-pq(N-1)(1+\xi-\alpha)b]-x(1+\xi-\alpha)[-(1-pq)\\
&B(N-1)(\alpha-x-\xi)-pq(N-1)(\alpha-x-\xi)h-rc/N+c\\
&+G-pq(N-1)b(1+\xi-\alpha)],\\
\dot{\xi}&=(1+\xi-\alpha)(\alpha-\xi)[-(1-pq)B(N-1)(\alpha-x-\xi)\\
&-pq(N-1)(\alpha-x-\xi)h-rc/N+c+G-pq(N-1)b(1+\xi-\alpha)]\\
&-x(1+\xi-\alpha)[G-pq(N-1)b(1+\xi-\alpha)].
\end{aligned}
\right.
\end{eqnarray*}
We further let $M$ be a matrix whose columns are the eigenvectors of $J(0, 1-\alpha, \alpha)$ which can be written as
\begin{equation*}
M=\begin{bmatrix}
\frac{a_{11}-a_{22}}{a_{21}} &&& 0\\
1 &&& 1
\end{bmatrix}.
\end{equation*}
Then we have
\begin{equation*}
M^{-1}J(0, 1-\alpha, \alpha)M=\begin{bmatrix}
0 &&& 0\\
0 &&& a_{22}
\end{bmatrix}.
\end{equation*}
We further take $[u \quad v]^{T} = M^{-1}[x \quad \xi]^{T}$, and thus we have $u = x / \theta$ and
$v = \xi-x / \theta$
where we use the notation $\theta=\frac{a_{11}-a_{22}}{a_{21}}$. It leads to
\begin{eqnarray*}
\dot{u}&=&u(1-\theta u)[-(u+v)pq(N-1)b]-u(1+u+v-\alpha)\\
&\times&[-(1-pq)B(N-1)(\alpha-\theta u-u-v)-pq(N-1)(\alpha-\theta u-u-v)h\\
&-&rc/N+c-pq(N-1)b(u+v)].
\end{eqnarray*}
Using the center manifold theorem, we have that $v = e(u)$ is a center manifold. We assume that $e(u) = O(|u|^{2})$, which yields the reduced system
\begin{eqnarray*}
\dot{u}&=-u^{2}pq(N-1)b+\theta u^{3}pq(N-1)b+O(|u|^{4}).
\end{eqnarray*}
Since $-pq(N-1)b\neq0$, we know that the fixed point $(u, e(u))=(0, 0)$ is unstable for the reduced system. Accordingly, the fixed point $(0, 1-\alpha, \alpha)$ is unstable.

(2) For $\alpha\leq0$ or $\alpha\geq1$, the boundary equilibrium point of the $DP$ edge does not exist. Then the system has only three fixed points in the parameter space, namely, $(0, 0, 1), (1, 0, 0),$ and $(0, 1, 0)$.

The fixed point $(0, 0, 1)$ and $(1, 0, 0)$ are both unstable since the largest eigenvalues of corresponding Jacobian matrices are both positive. If $\frac{rc}{N}-c-G+pq(N-1)b<0$, the fixed point $(0, 1, 0)$ is stable, while it is a saddle point and becomes unstable when $\frac{rc}{N}-c-G+pq(N-1)b>0$. Particulary, when $rc/N-c-G+pq(N-1)b=0$ and $(1-pq)B+pqh-pqb\neq0$, it is unstable.

\section{Numerical examples}

\begin{figure*}[!t]
\begin{center}
\includegraphics[width=5in]{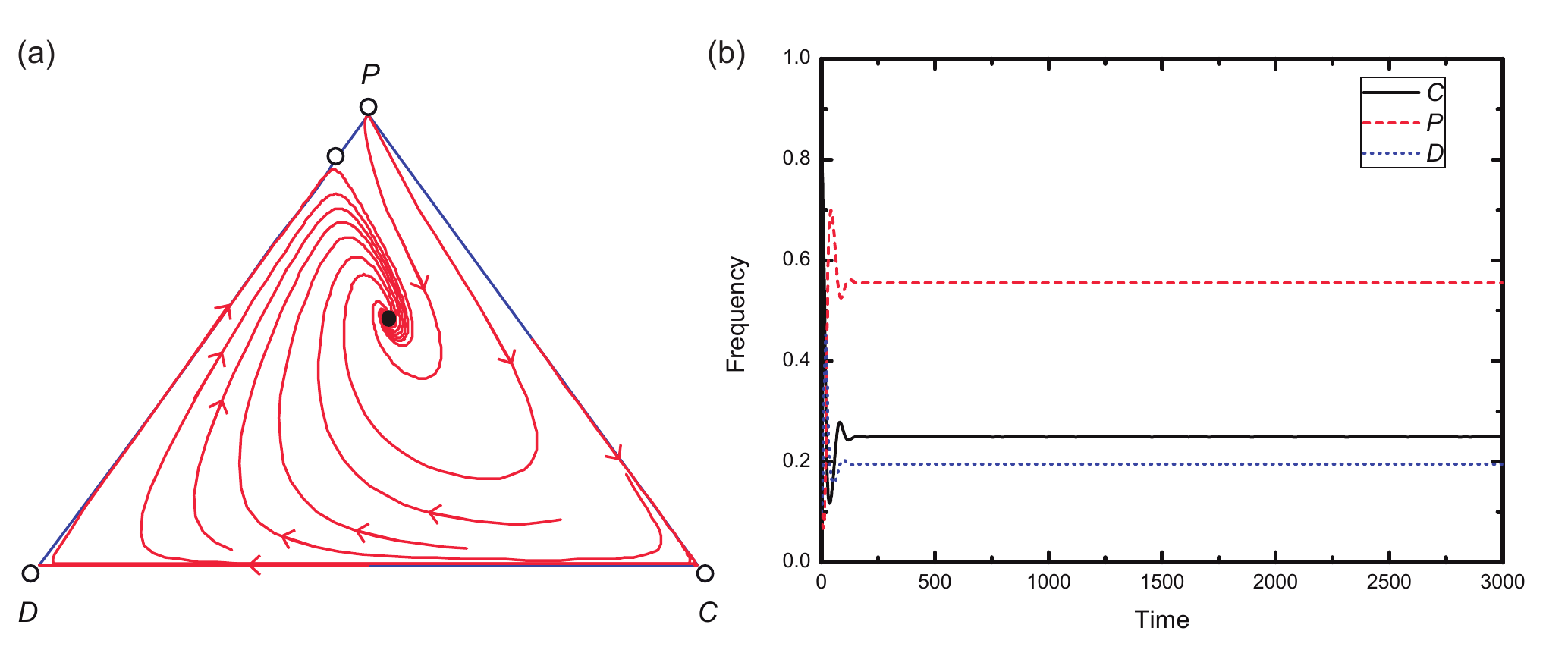}
\caption{The evolution of cooperation in an environment in which the expected loss of defectors is less than the expected bribe amount received by punishers. The triangle represents the state space where the actual fractions of cooperators, defectors, and punishers are denoted $\Delta=\{(x, y, z): x, y, z\geq0$ and $x+y+z=1\}$. Filled circle represents a stable fixed point whereas open circles represent unstable fixed points. Panel~(a) depicts that a stable interior equilibrium appears in the simplex $S_{3}$, which means that these three strategies coexist by maintaining a stable fraction in the population. Panel~(b) depicts the time series of the frequencies of three strategies $C$ (cooperators, black solid line), $D$ (defectors, blue dot line), and $P$ (pool punishers, red dash line). After an initial transient the frequencies of the three strategies eventually stabilize. Initial conditions: $(x, y, z)=(0.8, 0.1, 0.1).$ Parameters: $N=5$, $r=3$, $c=1$, $G=0.5$, $B=0.5$, $h=0.1, b=0.8, q=1$, and $p=0.8$.}\label{fig1}
\end{center}
\end{figure*}

We now provide some numerical examples to confirm the above theoretical analysis. We use the simplex $S_{3}=\{(x, y, z): x, y, z\geq 0, x+y+z=1\}$ to depict the state space of above three strategies. Accordingly, the three homogeneous states $C$ ($x=1$), $D$ ($y=1$), and $P$ ($z=1$) correspond to three corners of the simplex $S_{3}$. All of these are equilibrium points of the system (\ref{system1}). We first present numerical cases when the system (\ref{system1}) has an interior equilibrium point. When $\beta<1, \mu<1$, and $\beta+\mu<1$, there is an interior fixed point. From the theoretical analysis we know that its stability is determined by the relationship between the expected loss $(1-pq)B+pqh$ of defectors and the expected bribe amount $pqb$ received by punishers.

\subsection{Stable interior equilibrium point}

In the following we present a numerical example to confirm that the system can have a stable interior equilibrium point. As shown in Fig.~\ref{fig1}(a), we see that there exist five equilibria in the simplex $S_{3}$. The existing interior equilibrium point is stable and all interior orbits converge to this point, irrespective of the initial conditions. Besides, a boundary equilibrium point between all defectors and all punishers appears on the edge $DP$, which is unstable. Furthermore, the direction of the evolution on the $CP$ edge is from $P$ to $C$, and from $C$ to $D$ on the $CD$ edge. Figure~\ref{fig1}(b) depicts the time evolution of the frequencies of cooperators, defectors, and punishers. We can see that all mentioned strategists can coexist when the system reaches the steady state. Furthermore, the fraction of punishers is the highest, and the fraction of defectors is the lowest.

It is worth noting that all these results above are consistent with the theoretical results. In particular, when $(1-pq)B+pqh-pqb<0$, the interior equilibrium point is stable. And the boundary fixed point with $z=\alpha$ on $DP$ edge can appear when $0<\alpha<1$. According to our theoretical analysis, we know that it is unstable since $G-pq(N-1)(1-\alpha)b>0$ (see Theorem~\ref{T3.1} and Sec.~\ref{subsection3.2.1}).

\subsection{Hamiltonian system}
\begin{figure*}[!t]
\begin{center}
\includegraphics[width=5in]{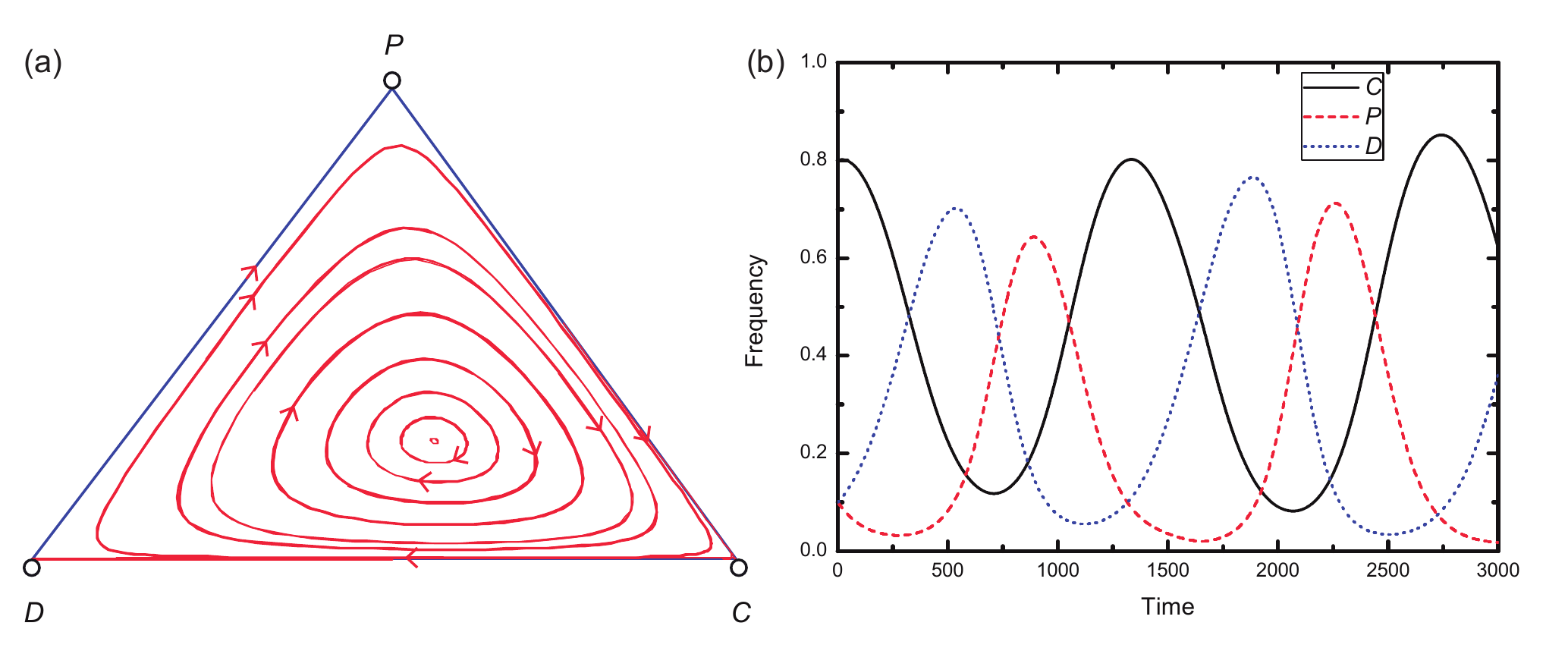}
\caption{The evolution of cooperation in a corrupt environment in which the expected loss of defectors and the expected bribe amount received by punishers are equal.
Panel~(a) depicts that there is an interior fixed point which is surrounded by periodic orbits in the simplex. Panel~(b) depicts a representative behavior when strategies are oscillating among $C$, $D$, and $P$ states. Initial conditions: $(x, y, z)=(0.8, 0.1, 0.1)$. Parameters: $N=5$, $r=3$, $c=1$, $G=0.5$, $B=0.5$, $h=0.4, b=0.4, q=1$, and $p=1$.}\label{fig2}
\end{center}
\end{figure*}

In this section we provide a numerical example to confirm that the system can become a Hamiltonian system. In Fig.~\ref{fig2}(a), we see that the existing interior equilibrium point is a center surrounded by periodic orbits and the strategy evolution trajectory displays a limit cycle. No boundary fixed points can be detected on the three edges. And the direction of evolution on the $DP$ edge is from $D$ to $P$, from $P$ to $C$ on the $PC$ edge, and from $C$ to $D$ on the $CD$ edge. Figure~\ref{fig2}(b) depicts that the frequencies of these three strategies display periodic oscillations in dependence on time, which is corresponding to the limit cycle shown in Fig.~\ref{fig2}(a). All these results are in agreement with the theoretical prediction, namely, when $(1-pq)B+pqh-pqb=0$ the interior fixed point is neutrally stable, and the dynamic system is Hamiltonian (see Theorem~\ref{theorem3.2} of Sec.~\ref{section3}).

The cyclical evolutionary scenario can be described as follows. If most players are cooperators in the group, it is better to become a defector due to the social dilemma. If defectors are prevalent, corrupt officials can get a lot of bribes from a group of defectors, and thus the number of punishers increases. If most players are punishers, the bribes from a few defectors are usually small enough to subvert cooperators dominance over punishers, and thus cooperators spread. If the number of cooperators increases sufficiently, then the original cooperation dilemma returns. It is worth noting that cooperative behavior can still emerge even in a completely corrupt environment. As we set that $p=q=1$, which means that all defectors are willing to offer a bribe to punishers, and all punishers are corrupt officials. It is inspiring to see that altruistic behavior can still be maintained due to the oscillations recurrent increase in cooperation.

\subsection{Heteroclinic cycle}
\begin{figure*}[!t]
\begin{center}
\includegraphics[width=5in]{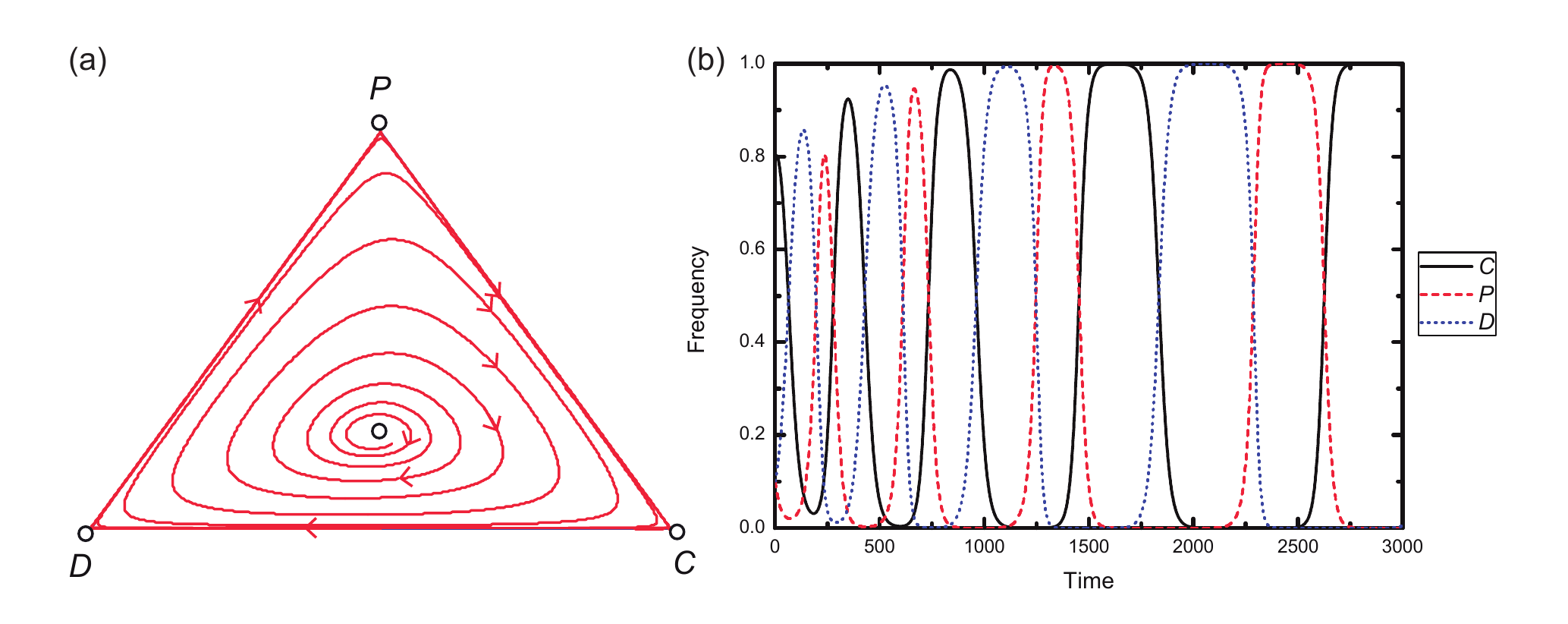}
\caption{The evolution of cooperation in an environment in which the expected loss of defectors is larger than the expected bribe amount received by punishers. Panel~(a) depicts that the interior equilibrium point turns into a repeller, and the interior curves of simplex $S_{3}$ coverage to the boundary of $S_{3}$. Panel~(b) depicts that the frequencies of these three strategies display growing periodic oscillations. Initial conditions: $(x, y, z)=(0.8, 0.1, 0.1).$ Parameters: $N=5$, $r=3$, $c=1$, $G=0.5$, $B=0.5$, $h=0.4, b=0.4, q=0.8$, and $p=1$.}\label{fig3}
\end{center}
\end{figure*}

We provide a numerical example to confirm that the heteroclinic cycle can exist in our equation system. Figure~\ref{fig3}(a) shows that there are four unstable fixed points in the simple $S_{3}$, and all interior curves of state space converge to the boundary of $S_{3}$. Besides, the direction of the evolution on the $DP$ edge is from $D$ to $P$, on the $PC$ edge is from $P$ to $C$, while on the $CD$ edge is from $C$ to $D$. Figure~\ref{fig3}(b) depicts that the frequencies of the three strategies are oscillating and the amplitudes are gradually growing, eventually forming periodic oscillations. This means that these three strategies can be cyclically dominant in the population and can thus avoid the collapse of cooperation. Note that these numerical results are in agreement with the theoretical analysis which predicts that the population converges to a stable heteroclinic cycle on the boundary of $S_3$ for $(1-pq)B+pqh-pqb>0$, $\frac{rc}{N}-c-G>\max\{-(N-1)[(1-pq)B+pqh],-pq(N-1)b\}$ and $r<N$ (Theorem~\ref{theorem3.4} of Sec.~\ref{section3}).

\subsection{Global stability of all-D state}
\begin{figure*}[!t]
\begin{center}
\includegraphics[width=5in]{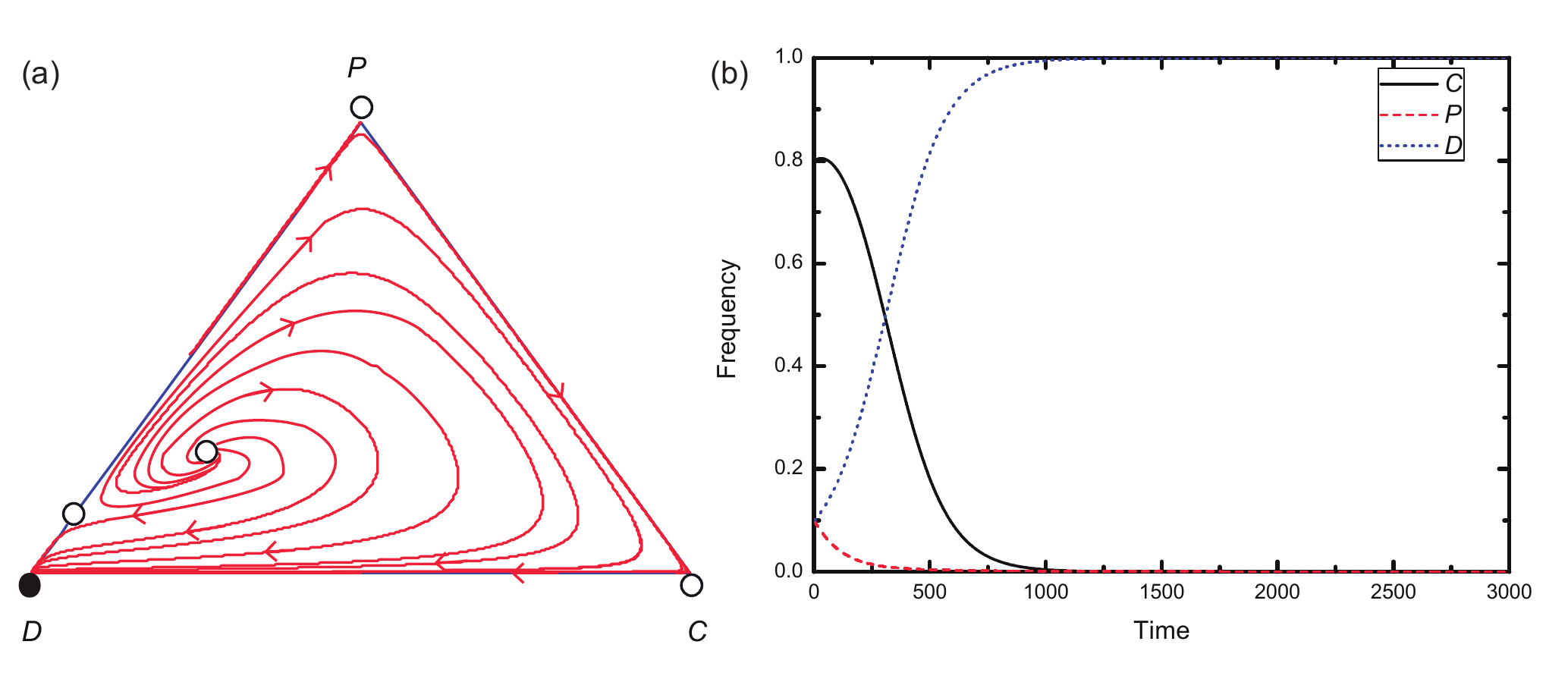}
\caption{The evolution of cooperation in an environment in which the expected loss of defectors is larger than the expected bribe amount received by punishers. Panel~(a) depicts that the defection has an evolutionary advantage over the two other strategies, and ultimately occupies the entire population. Panel~(b) illustrates the time evolution of cooperation, defection, and punishment in the whole population. Initial conditions: $(x, y, z)=(0.8, 0.1, 0.1).$ Parameters: $N=5$, $r=3$, $c=1$, $G=0.5$, $B=0.5$, $h=0.4, b=0.2, q=1$, and $p=1$.}\label{fig4}
\end{center}
\end{figure*}

Next we provide a numerical example to confirm that the equation system can converge to the all-$D$ state with global stability when the interior equilibrium point is unstable. As shown in Fig.~\ref{fig4}(a), there are five fixed points in the simplex $S_{3}$. All interior orbits coverage to the vertex $D$, which is a global stable equilibrium point. In the boundary case the unstable equilibrium point on the $DP$ edge reappears. Besides, defectors can always do better than cooperators on the $CD$ edge, and cooperators have more advantages than punishers on the remaining $CP$ edge. In Fig.~\ref{fig4}(b), we show how the frequencies of the three strategies evolve with time when the initial fractions of cooperators, defectors, and punishers are $0.8$, $0.1$, and $0.1$, respectively. Finally the fraction of defectors reaches one. This numerical calculation is in agreement with our theoretical prediction that the system converges on the all-$D$ state for $\frac{rc}{N}-c-G+pq(N-1)b<0$ (Sec.~\ref{subsection3.2.1}).

\subsection{Boundary equilibrium point with the coexistence of defectors and punishers}
\begin{figure*}[!t]
\begin{center}
\includegraphics[width=5in]{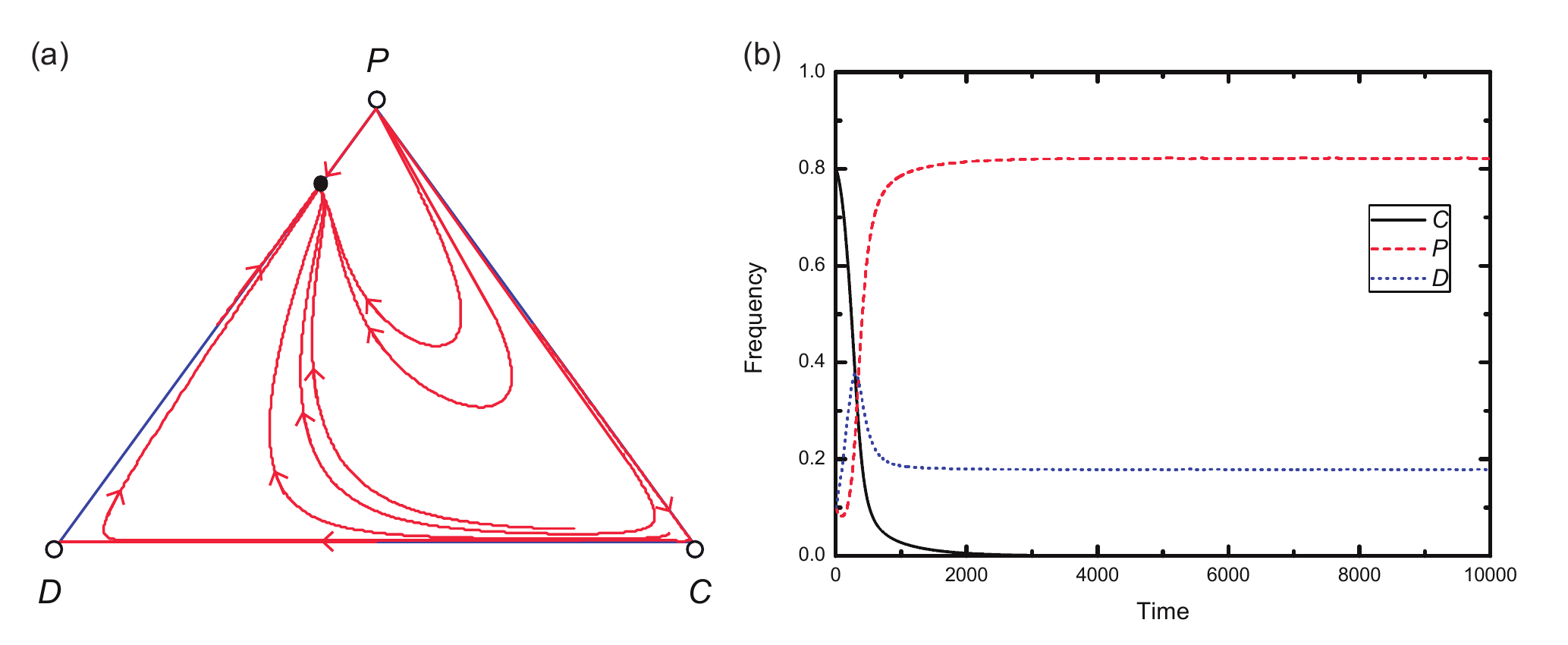}
\caption{The evolution of cooperation in a corrupt environment in which the expected loss of defectors is less than the expected bribe amount received by punishers. Panel~(a) depicts that there exists a stable boundary equilibrium point on $PD$ edge, and all trajectories starting from different initial conditions terminate onto this point. Panel~(b) illustrates how the fractions of cooperators, defectors, and punishers evolve in the population.
Initial conditions: $(x, y, z)=(0.8, 0.1, 0.1).$ Parameters: $N=5$, $r=3$, $c=1$, $G=0.5$, $B=0.5$, $h=0.1, b=0.8, q=1$, and $p=1$.}\label{fig5}
\end{center}
\end{figure*}

Finally we present a numerical example to confirm that there exists a stable boundary equilibrium point with the coexistence of defectors and punishers. In Fig.~\ref{fig5}(a), we show that the system converges to the stable boundary equilibrium point on the $DP$ edge, irrespective of the initial conditions, which means that defectors and punishers can coexist permanently in the population, while cooperators disappear. Besides, the evolutionary direction on $CD$ edge is from $C$ to $D$, while it is from $P$ to $C$ on the $CP$ edge. Figure~\ref{fig5}(b) shows that the fraction of cooperators gradually reduces to zero, even though their initial frequency is relatively high. In parallel the fraction of punishers finally reaches about $0.8$, while the fraction of defectors converges to about $0.2$. These numerical results validate our theoretical analysis based on which as long as the interior equilibrium does not exist, the boundary fixed point on the $DP$ edge is globally stable (Sec.~\ref{subsection3.2.2}).

\section{Conclusions}

In this paper, we have introduced probabilistic corrupt enforcers and violators into the PGG and investigated their consequence on the evolution of cooperation and punishment in infinite well-mixed populations. As a result, we have observed basically two different evolutionary scenarios. Namely, the competing strategies either coexist by forming stable time-dependent fractions or they dominate each other in a rock-scissors-paper game like manner. But for both cases the cooperative behavior can be well preserved. We have further identified the conditions in which the two dynamic behaviors can appear. We find that when the expected loss is less than the related bribe amount, a stable coexistence state among these three strategists can appear. In the alternative case, when these two quantity values are identical, the replicator dynamical system can be reduced to a Hamiltonian system. Here a center surrounded by closed orbits appears in the interior of the simplex $S_{3}$, which means that these three strategies are mutually restrained and exhibit periodic oscillations. When the expected loss is larger than the expected bribe amount, the interior fixed point is unstable, and two different evolutionary dynamic behaviors can emerge, namely, stable heteroclinic cycle and global attractor $D$. The former can guarantee that cooperators, defectors, and punishers are cyclically dominant in the population, while the latter will lead to the collapse of cooperation.

The key feature of our model is that corruption may emerge stochastically, either from the side of violators or from corrupt enforcers. This realistic assumption can be modeled in infinite well-mixed populations by using replicator equations. The study of such highly nonlinear equation systems is really challenging~\cite{Hauert02}, but we could manage a comprehensive and systematic theoretical analysis. We could identify all equilibrium points and characterize their stability by appropriately linearizing the equation system. In particular, when there are pure imaginary among the eigenvalues of the Jacobian we cannot determine the stability of an equilibrium point based only on the eigenvalues. Instead, we investigate its stability by utilizing the center manifold theorem~\cite{Carr81,Khalil96}. Furthermore, we not only prove that the system can exhibit the central limit cycle, but also theoretically confirm the existence and stability of heteroclinic cycle.

Recently, Huang et al.~\cite{Huang2018} investigated the effect of corruption on the evolution of cooperation and punishment in a hierarchical society which is divided into civilians and cops. Our present work, however, focuses on the evolutionary dynamics of cooperation, defection, and punishment in an integrated society when corruption is possible. Importantly, we consider that defectors probabilistically bribe enforcers to avoid punishment and meanwhile pool punishers probabilistically accept a bribe from defectors, which can reflect a realistic behavior manner of individuals in human societies. This behavior was studied in a population where individuals play the PGG with pool punishment~\cite{Sigmund2011Social}, and thus provides an alternative, but reasonable route to study the effects of corruption on collective actions of cooperation. Rather unexpectedly, the introduction of corruption can stabilize cooperation, and more interesting dynamic behaviors can be derived which may provide some implications for designing sanctioning strategies to support the evolution of cooperation.

We point out that although cooperative behavior can occur in a corrupt environment, the long-term existence of corruption will seriously
harm the economy and society as a whole. For example, one cross-cultural experiment involving thousands of people worldwide showed that
corruption not only deprives people of economic prosperity and growth, but also jeopardizes their intrinsic honesty~\cite{Salvi16}. We should stress that our present model does not explicitly consider a strategy chance of an anti-corruption control which could be an additional source of cooperation support. Such kind of extension could be the scope of future research. On the other hand, however, although powerful anti-corruption monitoring and sanctioning have been used to resolve the corruption problem, the effects of such measures are either transient or uncertain~\cite{Abdallah14,Lee2017,Muthukrishna2017,Verma2017,Huang2018}. For example, Muthukrishna et al.~\cite{Muthukrishna2017} found that anti-corruption strategies are effective under some conditions, but can further decrease public good provisioning when leaders are weak and the economic potential is poor even if these strategies are powerful. The failure of anti-corruption strategies is related to many independent factors, and the most important factor may be that the reason why enforcers are apt to fall into corruption is unclear. A recent experimental study has revealed the reason why the upper-class individuals behave more unethically than lower-class individuals is that they are more favorable attitudes toward greed~\cite{piff2014pnas}. Consequently, the task how to design efficient anti-corruption strategies for resisting the occurrence of corruption from the source remains an open and challenging question.

\section*{Acknowledgments}
This research was supported by the National Natural Science
Foundation of China (Grant Nos. 61976048 and 61503062), and by the Hungarian National Research Fund (Grant K-120785), and by the Science Strength Promotion Programme of University of Electronic Science and Technology of China.


\begin{thebibliography}{00}

\bibitem{Abdallah14}
S. Abdallah, R. Sayed, I. Rahwan, B. L. LeVeck, M. Cebrian, A. Rutherford, and J.
H. Fowler, Corruption drives the emergence of civil society, {\it J. Roy. Soc. Interface} {\bf11}
(2014) 20131044.
\bibitem{Marsan16}
G. Ajmone Marsan, N. Bellomo, and L. Gibelli, Stochastic evolutionary differential games
toward a systems theory of behavioral social dynamics, {\it Math. Models Methods Appl. Sci.} {\bf 26} (2016), 1051-1093.
\bibitem{Allen18}
J. M. Allen, A. C. Skeldon, and R. B. Hoyle, Social influence preserves cooperative strategies in the conditional cooperator public goods game on a multiplex network, {\it Phys. Rev. E} {\bf 98 }(2018) 062305.
\bibitem{Andreoni03}
J. Andreoni, W. Harbaugh, and L. Vesterlund, The carrot or the stick: Rewards, punishments, and cooperation,
{\it Am. Econ. Rev.} {\bf 93} (2003) 893--902.

\bibitem{Bellomo09}
N. Bellomo, H. Berestycki, F. Brezzi, and J. Nadal, Mathematics and complexity in
life and human sciences, {\it Math. Models Methods Appl. Sci.} {\bf 19} (2009) 1385-1389.

\bibitem{Bellomo11}
N. Bellomo and F. Brezzi, Mathematics and complexity in biological sciences, {\it Math. Models Methods Appl. Sci.} {\bf 21} (2011) 819-824.

\bibitem{Bellomo16}
N. Bellomo, F. Brezzi, and M. Pulvirenti, Modeling behavioral social systems, {\it Math. Models Methods Appl. Sci.} {\bf 27} (2017) 1--11.
\bibitem{Bellomo17}
N. Bellomo and S. Y. Ha, A quest toward a mathematical theory of the dynamics of swarms, {\it Math. Models Methods Appl. Sci.} {\bf 27} (2017) 745--770.
\bibitem{Burini16}
D. Burini, S. De Lillo and L. Gibelli, Collective learning modeling based on the kinetic theory
of active particles, {\it Phys. Life Rev.}, {\bf 16} (2016), 123-139.
\bibitem{Carr81}
J. Carr, {\it Applications of Center Manifold Theory}, (Springer Verlag, New York, 1981)
%
\bibitem{Chenxj2015}
X. Chen, T. Sasaki, {\AA}. Br\"{a}nnstr\"{o}m, and U. Dieckmann, First carrot, then stick: how the adaptive hybridization of incentives promotes cooperation, {\it J. Roy. Soc. Interface} {\bf 12} (2014) 20140935.

\bibitem{chen18}
X. Chen and A. Szolnoki, Punishment and inspection for governing the commons in a feedback-evolving game, {\it  PLOS Comp. Biol.} {\bf 14} (2018) e1006347.

\bibitem{chen2012risk}
X. Chen, A. Szolnoki, and M. Perc, Risk-driven migration and the collective-risk social dilemma, {\it Phys. Rev. E} {\bf 86} (2012) 036101.

\bibitem{chen2014probabilistic}
X. Chen, A. Szolnoki, and M. Perc, Probabilistic sharing solves the problem of costly punishment, {\it New J. Phys.} {\bf 16} (2014) 083016.

\bibitem{Dolfin14}
M. Dolfin and M. Lachowicz, Modeling altruism and selfishness in welfare dynamics:
The role of nonlinear interactions, {\it Math. Models Methods Appl. Sci.} {\bf 24} (2014) 2361-2381.
\bibitem{Dolfin17}
M. Dolfin, L. Leonida, and N. Outada, Modelling human behaviour in economics and social
science, Physics of Life Reviews, 22 (2017), 1-21.
%
\bibitem{fehr2003nature}
E. Fehr and U. Fischbacher, The nature of human altruism, {\it Nature} {\bf 425} (2003) 785--791.
%
\bibitem{Fehr02}
E. Fehr and S. G\"{a}chter, Altruistic punishment in humans, {\it Nature}, {\bf 415} (2002) 137.
%
\bibitem{Fu2008Reputation}
F. Fu, C. Hauert, M. A. Nowak, and L. Wang, Reputation-based partner choice promotes cooperation in social networks, {\it Phys. Rev. E} {\bf 78} (2008) 026117.

\bibitem{fu2010invasion}
F. Fu, M. A. Nowak, and C. Hauert, Invasion and expansion of cooperators in lattice populations: Prisoner's dilemma vs. snowdrift games, {\it J. Theor. Biol.} {\bf 266} (2010) 358--366

\bibitem{Hamilton1963The}
W. D. Hamilton, The evolution of altruistic behavior, {\it Am. Nat.} {\bf 97} (1963) 354-356.
\bibitem{hardin1968tragedy}
G. Hardin, The tragedy of the commons, {\it Science} {\bf 162} (1968) 1243--1248.
\bibitem{hauert2002volunteering}
C. Hauert, S. De Monte, J. Hofbauer, and K. Sigmund, Volunteering as red queen mechanism for cooperation in public goods games, {\it Science} {\bf296} (2002) 1129--1132.
%
\bibitem{Hauert02}
C. Hauert, S. De Monte, J. Hofbauer, and K. Sigmund, Replicator dynamics for optional public good games, {\it J. Theor. Biol.} {\bf 218} (2002) 187--194.
\bibitem{he2018amc}
N. He, X. Chen, and A. Szolnoki. Central governance based on monitoring and reporting solves the collective-risk social dilemma, {\it Appl. Math. Comput.} {\bf 347} (2019) 334--341.
\bibitem{Herrero15}
M. A. Herrero and J. Soler, Cooperation, competition, organization: The dynamics of interacting living populations, {\it Math. Models Methods Appl. Sci.} {\bf 25} (2015) 2407--2415.
\bibitem{hilbe18}
C. Hilbe, L. Schmid, J. Tkadlec, K. Chatterjee, and M. A. Nowak, Indirect reciprocity with private, noisy, and incomplete information, {\it Proc. Natl. Acad. Sci. USA} {\bf 115} (2018) 12241-12246.
\bibitem{hofbauer1998evolutionary}
J. Hofbauer and K. Sigmund, {\it Evolutionary Games and Population Dynamics}, (Cambridge University Press, 1998).
\bibitem{Huang2018}
F. Huang, X. Chen, and L. Wang, Evolution of cooperation in a hierarchical society with corruption control, {\it J. Theor. Biol.} {\bf 449} (2018) 60--72.
\bibitem{Khalil96}
H. K. Khalil, {\it Noninear Systems}, (Prentice-Hall, 1996)
\bibitem{Lee2017}
J. H. Lee, M. Jusup, and Y. Iwasa, Games of corruption in preventing the overuse of common-pool resources,
{\it J. Theor. Biol.} {\bf 428} (2017) 76--86.
\bibitem{Liu2017Competitions}
L. Liu, X. Chen, and A. Szolnoki, Competitions between prosocial exclusions and punishments in finite populations, {\it Sci. Rep.} {\bf 7} (2017) 46634.
\bibitem{Liu2018Chaos}
L. Liu, S. Wang, X. Chen, and M. Perc, Evolutionary dynamics in the public goods games with switching between punishment and exclusion, {\it Chaos}, {\bf 28} (2018) 103105.

\bibitem{milinski2002reputation}
M. Milinski, D. Semmann, and H. J. Krambeck, Reputation helps solve the `tragedy of the commons', {\it Nature} {\bf 415} (2002) 424.
\bibitem{Muthukrishna2017}
M. Muthukrishna, P. Francois, S. Pourahmadi, and J. Henrich, Corrupting cooperation and how anti-corruption strategies may backfire, {\it Nat. Hum. Behav.} {\bf 1} (2017) 0138.
\bibitem{nowak2006five}
M. A. Nowak, Five rules for the evolution of cooperation, {\it Science} {\bf 314} (2006) 1560--1563.
\bibitem{Nowak04S}
M. A. Nowak and K. Sigmund, Evolutionary dynamics of biological games, {\it Science} {\bf 303} (2004) 793-799.
\bibitem{Pareschi2013}
L. Pareschi and G. Toscani, {\it Interacting multiagent systems: kinetic equations and Monte
Carlo methods} (Oxford University Press, Oxford, 2013).
%
\bibitem{park18}
J. Park, Balancedness among competitions for biodiversity in the cyclic structured three species system, {\it Appl. Math. Comput.} {\bf 320} (2018) 425--436.
%
\bibitem{Perc2017Statistical}
M. Perc, J. J. Jordan, D. G. Rand, and Z. Wang,  S. Boccaletti, and A. Szolnoki, Statistical physics of human cooperation, {\it Phys. Rep.} {\bf 687} (2017) 1--51.
\bibitem{Perc2008Social}
M. Perc and A. Szolnoki, Social diversity and promotion of cooperation in the spatial prisoner's dilemma game, {\it Phys. Rev. E} {\bf 77} (2008) 011904.
\bibitem{Perc2010Coevolutionary}
M. Perc and A. Szolnoki, Coevolutionary games--a mini review, {\it BioSystems} {\bf 99} (2010) 109--125.
\bibitem{piff2014pnas}
P. K. Piff, D. M. Stancato, S. C\^{o}t\'{e}, R. Mendoza-Denton, and D. Keltner, Higher social class
predicts increased unethical behavior, {\it Proc. Natl. Acad. Sci. USA} {\bf 109} (2012) 4086-4091.
\bibitem{rockenbach2006efficient}
B. Rockenbach and M. Milinski, The efficient interaction of indirect reciprocity and costly punishment,
{\it Nature} {\bf 444} (2006) 718--723.
\bibitem{Salvi16}
S. Salvi, Corruption corrupts: Society-level rule violations affect individuals' intrinsic honesty,
{\it Nature} {\bf 531} (2016) 456--457.
\bibitem{santos2011risk}
F. C. Santos, C. Francisco, and J. M. Pacheco, Risk of collective failure provides an escape from the tragedy of the commons, {\it Proc. Natl. Acad. Sci. USA} {\bf 108} (2011) 10421--10425.

\bibitem{santos2012role}
F. C. Santos, F. L. Pinheiro, T. Lenaerts, and J. M. Pacheco, The role of diversity in the evolution of cooperation, {\it J. Theor. Biol.} {\bf 299} (2012) 88--96.

\bibitem{santos2008social}
F. C. Santos, M. D. Santos, and J. M. Pacheco, Social diversity promotes the emergence of cooperation in public goods games, {\it Nature} {\bf 454} (2008) 213--216.
%
\bibitem{Santos12}
F. C. Santos, V. V. Vasconcelos, M. D. Santos, P. N. B. Neves, and J. M. Pacheco, Evolutionary dynamics of climate change under collective-risk dilemmas, {\it Math. Models Methods Appl. Sci.} {\bf 22} (2012) 1140004.

\bibitem{Sasaki2012The}
T. Sasaki, A. Br\"{a}nnstr\"{o}m, U. Dieckmann, and K. Sigmund, The take-it-or-leave-it option allows small penalties to overcome social dilemmas, {\it Proc. Natl. Acad. Sci. USA} {\bf 109} (2012) 1165--1169,
\bibitem{sasaki2013evolution}
T. Sasaki and S. Uchida, The evolution of cooperation by social exclusion, {\it Proc. R. Soc. B} {\bf 280} (2012) 20122498.
%
\bibitem{Sigmund2010nature}
K. Sigmund, H. De Silva, A. Traulsen, and C. Hauert, Social learning promotes institutions for governing the commons, {\it Nature} {\bf 466} (2010) 861--863.
%
\bibitem{Sigmund2011Social}
K. Sigmund, C. Hauert, A. Traulsen, and H. De Silva, Social Control and the Social Contract: The Emergence of Sanctioning Systems for Collective Action, {\it Dyn. Games Appl.} {\bf 1} (2011) 149--171.

%
\bibitem{szolnoki2017alliance}
A. Szolnoki and X. Chen, Alliance formation with exclusion in the spatial public goods game, {\it Phys. Rev. E} {\bf 95} (2017) 052316.
\bibitem{Szolnoki2018}
A. Szolnoki and X. Chen, Competition and partnership between conformity and payoff-based imitations in social dilemmas, {\it New. J. Phys.} {\bf 20} (2018) 093008.

\bibitem{szolnokipre18}
A. Szolnoki and X. Chen, Reciprocity-based cooperative phalanx maintained by overconfident players, {\it Phys. Rev. E} {\bf 98} (2018) 022309.

%
\bibitem{Szolnoki2009Resolving}
A. Szolnoki and M. Perc, Resolving social dilemmas on evolving random networks, {\it EPL} {\bf 86} (2009) 30007.
\bibitem{Szolnoki2012Evolutionary}
A. Szolnoki and M. Perc, Evolutionary advantages of adaptive rewarding, {\it New J. Phys.} {\bf 14} (2012) 93016.
\bibitem{Szolnoki2013Effectiveness}
A. Szolnoki and M. Perc, Effectiveness of conditional punishment for the evolution of public cooperation, {\it J. Theor. Biol.} {\bf 325} (2013) 34--41.
\bibitem{Szolnoki2007Cooperation}
A. Szolnoki and G. Szab{\'o}, Cooperation enhanced by inhomogeneous activity of teaching for evolutionary Prisoner's Dilemma games, {\it EPL} {\bf 77} (2007) 30004.
\bibitem{Szolnoki2011Phase}
A. Szolnoki, G. Szab{\'o}, and M. Perc, Phase diagrams for the spatial public goods game with pool punishment, {\it Phys. Rev. E} {\bf 83} (2011) 036101.
%

\bibitem{Tanimoto07}
J. Tanimoto and H. Sagara, Relationship between dilemma occurrence and the existence of a weakly dominant strategy in a two-player symmetric game, {\it BioSystems} {\bf 90} (2007) 105-114.
\bibitem{Vasconcelos13}
V. V. Vasconcelos, F. C. Santos, and J. M. Pacheco, A bottom-up institutional approach to cooperative governance of risky commons, {\it Nat. Clim. Change} {\bf 25} (2013) 797.
\bibitem{Vasconcelos15}
V. V. Vasconcelos, F. C. Santos, and J. M. Pacheco, Cooperation dynamics of polycentric climate governance, {\it Math. Models Methods Appl. Sci.} {\bf 25} (2015) 2503--2517.
\bibitem{Verma2017}
P. Verma, A. K. Nandi, and S. Sengupta, Bribery games on inter-dependent regular networks, {\it Sci. Rep.} {\bf 7} (2017) 42735.
\bibitem{Verma2015}
P. Verma and S. Sengupta, Bribe and punishment: An evolutionary game-theoretic analysis of bribery, {\it PLoS ONE} {\bf 10} (2015) e0133441.

\bibitem{wang2018}
Q. Wang, N. He, and X. Chen, Replicator dynamics for public goods game with resource allocation in large populations, {\it Appl. Math. Comput.} {\bf 328} (2018) 162--170.
\bibitem{Wu18}
T. Wu, F. Fu, and L. Wang, Coevolutionary dynamics of aspiration and strategy in
spatial repeated public goods games, {\it New J. Phys. } {\bf 20} (2018) 063007.
\bibitem{Yang2018}
H-X. Yang and X. Chen, Promoting cooperation by punishing minority, {\it Appl. Math. Comput.} {\bf 316} (2018) 460--466.
\end{thebibliography}
\end{document}